\documentclass[sigconf]{acmart}
\renewcommand\footnotetextcopyrightpermission[1]{} 
\AtBeginDocument{%
  \providecommand\BibTeX{{%
    \normalfont B\kern-0.5em{\scshape i\kern-0.25em b}\kern-0.8em\TeX}}}

\acmConference[]{}

\usepackage{soul}
\usepackage{url}
\usepackage[utf8]{inputenc}
\usepackage{caption}
\usepackage{graphicx}
\usepackage{amsmath}
\usepackage{booktabs}
\usepackage{multirow}
\usepackage{tabularx}
\usepackage{enumitem}
\usepackage{float}
\usepackage{subfloat}
\usepackage{subfig}
\graphicspath{{figures/}}
\usepackage[ruled,linesnumbered]{algorithm2e}

\newcolumntype{L}[1]{>{\raggedright\let\newline\\\arraybackslash\hspace{0pt}}m{#1}}
\newcolumntype{C}[1]{>{\centering\let\newline\\\arraybackslash\hspace{0pt}}m{#1}}
\newcolumntype{R}[1]{>{\raggedleft\let\newline\\\arraybackslash\hspace{0pt}}m{#1}}

\newcolumntype{R}{>{\raggedleft\arraybackslash}X}
\newcolumntype{C}{>{\centering\arraybackslash}X}
\usepackage{array}
\begin{document}

\title{Image Hashing by Minimizing Discrete Component-wise Wasserstein Distance}



\author{Khoa D. Doan}
\email{khoadoan@vt.edu}
\affiliation{%
  \institution{Virginia Tech}
  \city{Arlington}
  \country{USA}
}

\author{Saurav Manchanda}
\email{manch043@umn.edu}
\affiliation{%
  \institution{University of Minnesota}
  \city{Twin-Cities}
  \country{USA}}

\author{Sarkhan	Badirli}
\email{s.badirli@gmail.com}
\affiliation{%
  \institution{Purdue University}
  \city{Indiana}
  \country{USA}
}

\author{Chandan K. Reddy}
\email{reddy@cs.vt.edu}
\affiliation{%
  \institution{Virginia Tech}
  \city{Arlington}
  \country{USA}
}





\begin{abstract}
Image hashing is one of the fundamental problems that demand both efficient and effective solutions for various practical scenarios. Adversarial autoencoders are shown to be able to implicitly learn a robust, locality-preserving hash function that generates balanced and high-quality hash codes. However, the existing adversarial hashing methods are inefficient to be employed for large-scale image retrieval applications. Specifically, they require an exponential number of samples to be able to generate optimal hash codes and a significantly high computational cost to train. In this paper, we show that the high sample-complexity requirement often results in sub-optimal retrieval performance of the adversarial hashing methods. To address this challenge, we propose a new adversarial-autoencoder hashing approach that has a much lower sample requirement and computational cost. Specifically, by exploiting the desired properties of the hash function in the low-dimensional, discrete space, our method efficiently estimates a better variant of Wasserstein distance by averaging a set of easy-to-compute one-dimensional Wasserstein distances. The resulting hashing approach has an order-of-magnitude better sample complexity, thus better generalization property, compared to the other adversarial hashing methods. In addition, the computational cost is significantly reduced using our approach. We conduct experiments on several real-world datasets and show that the proposed method outperforms the competing hashing methods, achieving up to 10\% improvement over the current state-of-the-art image hashing methods. The code accompanying this paper is available on Github\footnote{\url{https://github.com/khoadoan/adversarial-hashing}}. 

\end{abstract}




\maketitle

\section{Introduction}

The rapid growth of visual data, especially images, brings many challenges to the problem of finding similar items. {Exact similarity search}, which aims to exhaustively find all relevant images, is often impractical due to its computational complexity. This is due to the fact that a complete linear scan of all the images in such massive databases is not feasible, especially when the database contains millions (or billions) of items. Hashing is an approximate similarity search method which provides a principled approach for web-scale databases. In hashing, high-dimensional data points are projected onto a much smaller locality-preserving \textit{binary} space via a hash function $f: x \rightarrow \{0, 1\}^m$, where $m$ is the dimension of the binary space. \textit{Approximate search for similar images can be efficiently performed} in this binary space using Hamming distance  \cite{leskovec2014mining}. Furthermore, the compact binary codes are storage-efficient. 


The existing hashing methods can be broadly grouped into supervised and unsupervised hashing.  Although supervised hashing offers a superior performance, unsupervised hashing is more suitable for large databases because it learns the hash function without any labeled data. One of the widely used unsupervised hashing techniques is Locality Sensitive Hashing (LSH)~\cite{leskovec2014mining}. Image hashing can also be sub-categorized as shallow hashing methods, such as Spectral Hashing (SpecHash)~\cite{weiss2009spectral} and Iterative Quantization (ITQ)~\cite{gong2013iterative}, and deep hashing methods, such as SSDH and DistillHash~\cite{yang2018semantic,yang2019distillhash,ghasedi2018unsupervised}. 


Even though existing methods have shown some reasonable performance improvements in several image-hashing applications, they have two main drawbacks: (1) their objective functions are heuristically constructed without a principled characterization of the goodness of the hash codes, and (2) the gap between the desired, discrete solution and the relaxed, continuous solution is minimized heuristically with explicit constraints. The latter increases the time to tune the additional hyperparameters of the models.  The authors of \cite{doan2019adversarial,doan2020efficient} show that employing adversarial autoencoders for hashing avoids these explicitly constructed constraints. Their hashing model implicitly learns the hash functions by adversarially match the hash functions' output with a target, supposedly optimal discrete prior. This removes the need for time-consuming hyperparameter tuning. Training adversarial autoencoders by minimizing the Jensen-Shannon or Wasserstein distance is, however, difficult, especially when the dimension of the latent space increases. The scaling difficulty of the adversarial autoencoders may be related to one fundamental issue: the generalization property of matching distributions. The authors of~\cite{arora2017generalization} show that Jensen-Shannon divergence and Wasserstein distance do not generalize, in a sense that the generated distribution cannot converge to the target distribution without an exponential number of samples. \textit{Without good generalization, the retrieval performance can be sub-optimal}.

To address these aforementioned challenges, we propose a novel unsupervised Discrete Component-wise Wasserstein Autoencoder (DCW-AE) model for the image hashing problem. The proposed model implicitly learns the optimal hash function using a novel and efficient divergence minimization framework. The main contributions of the paper are as follows:

\begin{itemize}[leftmargin=*]
    \item Demonstrate that the ability to match the distribution of the output of the learned hash function to the target discrete distribution matching is closely related to the retrieval performance. Specifically, employing a distance with an easier convergence to the target distribution (called generalization) results in better retrieval performance. To this end, the existing Wasserstein-based Adversarial Autoencoders have poor generalization; thus they have a sub-optimal retrieval performance.  
    \item Propose a novel, efficient approach to learn the hash functions by employing a more generalizable variant of the Wasserstein distance, that leverages the discrete properties of hashing. It has an order-of-magnitude better generalization property and an order-of-magnitude more efficient computation than the existing Wasserstein-based hashing methods.
    
    \item Demonstrate the superiority of the proposed model over the state-of-the-art hashing techniques on various widely used real-world datasets using both quantitative and qualitative performance analysis.
\end{itemize}

The rest of the paper is organized as follows. We discuss the related work in Section \ref{relatedworks}. In Section \ref{proposed_method}, we describe the details of the proposed method. Finally, we present quantitative and qualitative experimental results in Section \ref{experiments} and conclude our discussion in Section \ref{conclusions}.

\section{Related Work} \label{relatedworks}

In this section, we begin by discussing the related work in the image hashing domain, with main focus towards the motivation behind adversarial autoencoders. Then, we continue our discussion on  adversarial learning and their limitations, especially on their generalization property when matching to the target distribution (or sample complexity requirement).

\subsection{Image Hashing}

Various supervised and unsupervised methods have been developed for hashing. Examples of supervised hashing methods include~\cite{shen2015supervised,yang2018supervised,ge2014graph,xia2014supervised,cao2018hashgan} and examples of unsupervised hashing methods include~\cite{huang2017unsupervised,lin2016learning,huang2016unsupervised,do2016learning,he2013k,heo2012spherical,gong2013iterative,weiss2009spectral,salakhutdinov2009semantic,yang2019distillhash,dizaji2018unsupervised}. While supervised methods demonstrate a superior performance over unsupervised ones, they require human-annotated datasets. Annotating massive-scale datasets, which are common in the image hashing domain, is an expensive and tedious task. Furthermore, besides the train/test distribution-mismatch problem, supervised methods easily get stuck in bad local optima when labeled data are limited. Thus, exploring the unsupervised hashing techniques is of great interest, especially in the image-hashing domain. 

Hashing methods can also be categorized as either data independent or dependent. One of the most popular data-independent hashing technique is LSH~\cite{leskovec2014mining}. Data-dependent hashing includes popular methods such as SpecHash~\cite{weiss2009spectral} and ITQ~\cite{gong2013iterative}. Data-dependent hashing demonstrates a significant increase in retrieval performance because it considers the data distribution. Hashing methods can also be categorized as shallow~\cite{gionis1999similarity,weiss2009spectral,heo2012spherical} and deep hashing~\cite{do2016learning,lin2016learning,huang2017unsupervised,yang2019distillhash,dizaji2018unsupervised}. The deep hashing methods can learn non-linear hash functions and have shown superiority over the shallow approaches. 

In general, the existing hashing methods learn the hash functions by minimizing the following training objective:
\begin{equation}
    \min_{f} E_{x \sim D_x} L(x, f(x)) + E_{x \sim D_x} \sum_{k} \lambda_k \times H_k(f(x))
\end{equation}

where $D_x$ is the data distribution, $L(x, f(x))$ is the locality-preserving loss of the hash function $f(x)$ and $H_k(f(x))$ is a hashing constraint with $\lambda_k$ as its corresponding weight (thus a hyperparameter to tune). 

The authors of \cite{doan2019adversarial,doan2020efficient} show that by matching the latent space of the autoencoder with an optimal discrete prior, we can implicitly learn a good hash function $f$ while simultaneously satisfying the constraints $H_k$. However, their adversarial methods are unstable in practice and do not show a good generalization property. In Section~\ref{experiments}, we will show that poor generalization results in a sub-optimal retrieval performance when the model is trained with stochastic optimization techniques such as Stochastic Gradient Descent (SGD).

\subsection{Adversarial Learning}
\begin{figure*}[!hbpt]
\centering
\includegraphics[width=6in]{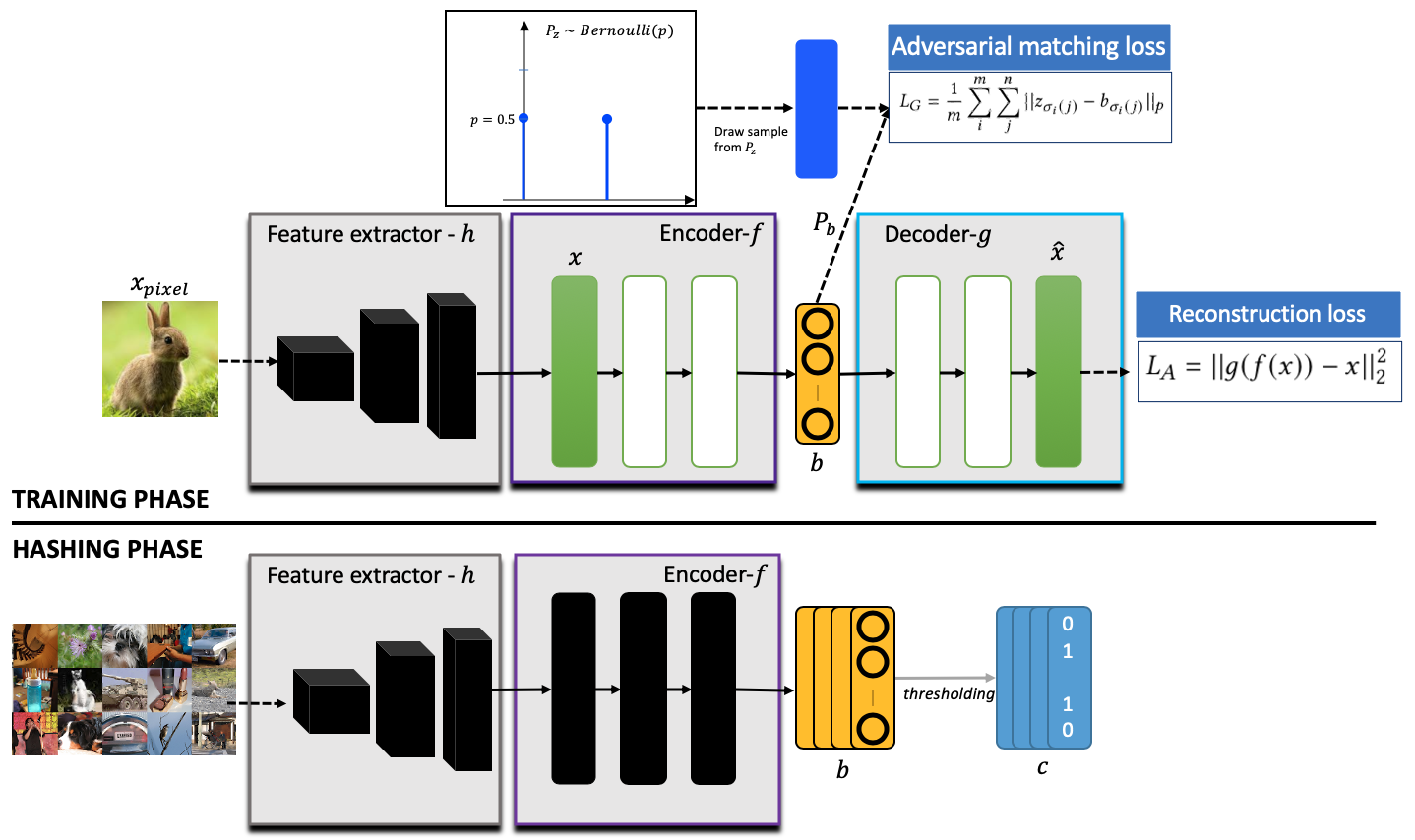}
\caption{The network architecture of the proposed DCW-AE model. During the training phase, the parameters of $f$ and $g$ are trained, while the parameters of $h$ are either fixed or trained. During the hashing phase, the vector $b$ is thresholded to obtain the hash code $c$.}
\label{fig:IDAE-network}
\end{figure*}

Generative Adversarial Network (GAN) has recently gained popularity due to its ability to generate realistic samples from the data distribution~\cite{goodfellow2014generative}. A prominent feature of GAN is its ability to ``implicitly'' \textit{match} output of a deep network to a pre-defined distribution using the adversarial training procedure. Furthermore, adversarial learning has been leveraged to regularize the latent space, as it helps in learning the intrinsic manifold of the data~\cite{makhzani2015adversarial,doan2019adversarial}. For example, the adversarially trained autoencoders can learn a smooth manifold of the data in the low-dimensional latent space~\cite{makhzani2015adversarial}.  However, training the adversarial autoencoders remains challenging and inefficient because of the alternating-optimization procedure (minimax game)  between the generator and the discriminator. For instance, the work presented in \cite{doan2019adversarial} employs the original minimax game~\cite{goodfellow2014generative}, which suffers from mode-collapse and vanishing gradient~\cite{arjovsky2017wasserstein}. Moreover, in the minimax optimization, the generator's loss fluctuates during training instead of ``descending'', making it extremely challenging to know when to stop training the model. 

Wasserstein-based adversarial methods overcomes a few of these limitations (specifically, mode-collapse and vanishing gradient)~\cite{arjovsky2017wasserstein,tolstikhin2017wasserstein}. However, because they employ the Kantorovich-Rubinstein dual, the minimax game still exists between the generator and the critic. On the other hand, the work in \cite{iohara2019generative,doan2020efficient} directly estimates the Wasserstein distance by solving the Optimal Transport (OT) problem. Solving the OT has two main challenges. Firstly, its computation cost is $O(N^{2.5} log(Nd))$ where $N$ is the number of data points and $d$ is the dimension of the data points. This is expensive. Secondly, the OT-estimate of the Wasserstein distance requires an exponential number of samples to generalize (or to achieve a good estimate of the distance)~\cite{arora2017generalization}. In practice, both the high-computational cost and exponential-sample requirement make the OT-based adversarial methods very inefficient. 


\sloppy Sliced Wasserstein Distance (SWD), on the other hand, approximates the Wasserstein distance by averaging the one-dimensional Wasserstein distances of the data points when they are projected onto many random, one-dimensional directions ~\cite{deshpande2018generative}. SWD is more generalizable than the OT, with polynomial sample complexity~\cite{deshpande2019max}. Furthermore, the SWD estimate has a computational cost of $O(N_{\omega}N \log(Nd))$, hence its  computational complexity is better than that of the OT \textit{only when $N_{\omega}$ is smaller than $N^{1.5}$}. Nevertheless, in the high dimensional space, it becomes very likely that many random, one-dimensional directions do not lie on the manifold of the data. In other words, along several of these directions, the projected distances are close to zero. Consequently, in practice, the number of random directions $N_{\omega}$ is often larger than $N^{1.5}$. For example, in~\cite{deshpande2018generative}, for a mini-batch size of $64$, SWD needs $N_{\omega}=10,000$ projections, which is significantly larger than $64^{1.5}$, to generate good visual images. To address this problem, Max-SWD finds the best direction and estimates the Wasserstein distance along this direction~\cite{deshpande2019max}. 

In this paper, we address the limitations of these GAN-based approaches by robustly and efficiently minimizing a novel variant of the Wasserstein distance. By carefully studying the properties of the target distribution in hashing, the proposed adversarial hashing method significantly more efficient than both the existing OT-based and SWD-based approaches.
 

\section{Proposed Method} \label{proposed_method}
\subsection{Problem statement}

Given a data set $X = \{x^{(1)},x^{(2)},...,x^{(N)}\}$ of $N$ images where $x^{(i)} \in R^d$, the goal of unsupervised hashing is to learn a hash function $f: x \rightarrow b$ that can generate binary hash code $b \in \{0, 1\}^m$ of the image $x$. $m$ denotes the length of the hash code $b$ and it is typically much smaller than $d$.

\begin{table}[!h]
\caption{Notations used in this paper.}
\begin{tabular}{|p{0.5in}|p{2.5in}|}
\hline
    \textbf{Notation}  & \textbf{Description} \\ \hline \hline
    $x_{pixel}, x$  	& the pixel vector  and the extracted feature vector of the image, respectively \\ \hline
    $h, f, g$  	& feature extractor, encoder, and decoder, respectively. \\ \hline
    $W_h, W_f, W_g$  	& parameters of the feature extractor, encoder, and decoder, respectively. \\ \hline
    $d$ & dimension of the data \\ \hline
    $m$ & dimension of the discrete space \\ \hline
    $c, b$ & discrete code and its continuous representation. \\ \hline
    $W, \hat{W}$ & Wasserstein distance and its empirical estimate, respectively \\ \hline
    $z$ & sample from the discrete prior $P_z$ \\ \hline
    $P_z, P_b$ & distributions of $z$ and $b$, respectively \\ \hline
    $\mathcal{D}, \mathcal{F}$ & empirical samples of $P_z$ and $P_b$, respectively. \\ \hline
    $N$ & number of samples. \\ \hline
    $\omega_k$ & vector that defines the projection onto a random one-dimensional direction. \\ \hline
    $N_{\omega}$ & number of random projections \\ \hline
    $L_A$ & autoencoder's reconstruction loss \\ \hline
    $L_G$ & adversarial matching loss (also called the distributional distance). \\ \hline
    
\end{tabular}
\label{tbl:notations}
\end{table}

\subsection{Network architecture}


We propose the DCW-AE network. Figure~\ref{fig:IDAE-network} shows the architecture of DCW-AE. Similar to the existing image hashing approaches~\cite{yang2019distillhash,yang2018semantic}, we choose to employ a feature extractor, such as the VGG network~\cite{simonyan2014very}, and represent an image by its extracted feature vector $x$. In particular, the feature extractor is defined as the function $h: x_{pixel} \rightarrow x$, where $x_{pixel}$ is the pixel-representation of the image. Note that, $h$ can be trained in an end-to-end framework similar to that of ~\cite{doan2019adversarial}. However, in our paper, we choose to use the pretrained VGG-feature extractor $h$ and do not retrain its parameters. 


The encoder, represented by the function $f: x \rightarrow b$, computes the low-dimensional representation $b$. Given the feature vector $x$ of an image, the output $b = f(x)$ is represented by the $m$ independent probabilities $b_{i} = p(c_{i}=1 | x, W_f)$, where $W_f$ is the parameter of the encoder. To generate the hash codes, we simply compute $c_{i} = \textbf{1}_{[b_{i} > 0.5]}$. Note that the encoder $f$ is also the hash function for our purpose. We then regularize the posterior distribution of $b$, called $P_b$ with a predefined, discrete prior $P_z$ by minimizing their distributional distance $L_G$. The decoder, represented by the function $g: b \rightarrow x$, reconstructs the input, denoted as $\hat{x}$. We train our model by minimizing the reconstruction loss $L_A$.

In the following sections, we will discuss the details of the proposed method, especially the novel, alternative formulation of the Wasserstein distance calculation which significantly improves the retrieval results of the existing adversarial autoencoders. 

\subsection{Locality preservation of the hash codes}

The autoencoder is trained to minimize the mean-squared error between the input and the reconstructed output, as below:
\begin{equation}
    L_A = \frac{1}{N} \sum_j^{N} ||\hat{x}^{(j)} - x^{(j)}||_2^2 = \frac{1}{N} \sum_j^{N} ||g(f(x^{(j)})) - x^{(j)}||_2^2
    \label{eqn:reconstruction_loss}
\end{equation}

It is easy to show that minimizing the reconstruction loss $L_A$ is equivalent to preserving locality information of the data in the original input space~\cite{dai2017stochastic}. In other words, the proposed autoencoder model learns the hash function $f$ that preserves the original input locality. Our approach to preserve the locality of the input in the discrete space using an autoencoder is different from the approaches taken by SSDH~\cite{yang2018semantic} and DistillHash~\cite{yang2019distillhash}. These methods heuristically constructs the semantic, pairwise similarity matrix from the representation $x$. Consequently, the retrieval performance closely depends on the quality of the representation $x$ and the constructed similarity matrix. As we shall see in Section~\ref{experiments}, when a good feature extractor $h$ is not available, our approach significantly outperforms these methods.




\subsection{Implicit optimal hash function learning} \label{sec:ot_ae}

We regularize the the encoder's output $b$ to match a pre-defined binary prior. Specifically, we sample a vector $z$ as the real data. Each component of $z$ is independently and identically sampled from a one-dimensional Bernoulli distribution with a parameter $p$. The sampling procedure defines a distribution $P_z$ over $z$ while the encoder defines a distribution $P_b$ over the latent space $b$. The encoder, which is the generator in the GAN game, learns its parameters $W_f$ by minimizing the Wasserstein distance as follows:

\begin{equation}
    W(P_b, P_z) = \inf_{\gamma \in \Pi(P_b, P_z)} \int_{(b,z) \sim \gamma} p(b,z) d(b,z) db dz
\end{equation}

where $\Pi(P_b, P_z)$ is the set of all possible joint distributions of $b$ and $z$ whose marginals are $P_b$ and $P_z$, respectively, and $d(b, z)$ is the cost of transporting one unit of mass from $b$ to $z$. 


Given a finite, $N$-sample $\mathcal{F}=\{b^{(1)},b^{(2)},...,b^{(N)}\}$ from $P_b$ and a finite, $N$-sample $\mathcal{D}=\{z^{(1)},z^{(2)},...,z^{(N)}\}$ from $P_z$, one approach is to minimize the empirical Optimal Transport (OT) cost as follows:

\begin{equation}
\hat{W}(\mathcal{D}, \mathcal{F}) = \min_{W_f} \sum_i^N \sum_j^N M_{ij} d(b^{(i)}, z^{(j)}) = \min_{W_f} M \odot D,
\label{eqn:primal_empirical}
\end{equation}
where $M$ is the assignment matrix, $D$ is the cost matrix where $D_{ij}=d(b^{(i)}, z^{(j)})$ and $\odot$ is the Hadamard product. This Linear Programming (LP) program has the following constraints:

\begin{align}
\sum_i^N M_{ij} = 1, \forall j=1,...,N\\
\sum_j^N M_{ij} = 1, \forall i=1,...,N\\
M_{ij} \in \{0, 1\}, \forall i=1,...,N, \forall j=1,...,N
\end{align}





The best method of solving the OT's LP program has a cost of approximately $O(N^{2.5}\log(Nd))$~\cite{burkard2009assignment}, where $N$ is the number of examples. While it is entirely possible to implement this LP program in a Stochastic Gradient Descent (SGD) training for small mini-batch sizes~\footnote{https://github.com/gatagat/lap}, it is computationally expensive for larger $N$. As discussed in Section~\ref{relatedworks}, the OT has an exponential sample complexity. This requirement makes the OT less suitable in practice, where large mini-batches are necessary for the models to perform well.

\subsection{Discrete Component-wise Wasserstein Distance}~\label{sec:dcw_ae}
\begin{figure}[!htbp]
 \centering
 \renewcommand{\thesubfigure}{a}
 \subfloat[$P_b$ is far from $P_z$. \label{subfig:random_vs_component_wise_projections_a}]{%
   \includegraphics[width=0.23 \textwidth]{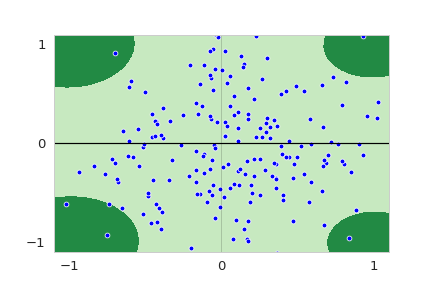}
 }
 \hfill
 \renewcommand{\thesubfigure}{b}
 \subfloat[1-D $\hat{W}$ along different projection angles (in degrees) when $P_b$ is far from $P_z$. \label{subfig:random_vs_component_wise_projections_b}]{%
   \includegraphics[width=0.23 \textwidth]{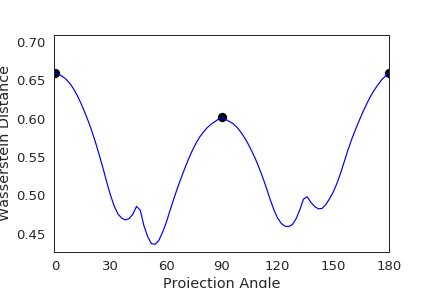}
 }
 \vfill
  \renewcommand{\thesubfigure}{c}
 \subfloat[$P_b$ is close to $P_z$.  \label{subfig:random_vs_component_wise_projections_c}]{%
   \includegraphics[width=0.23 \textwidth]{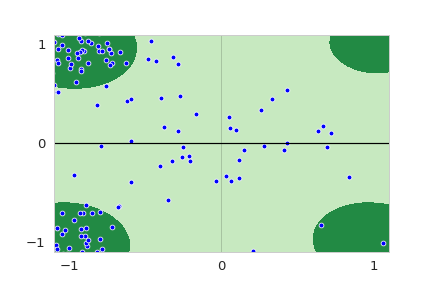}
 }
 \hfill
 \renewcommand{\thesubfigure}{d}
 \subfloat[1-D $\hat{W}$ along different projection angles (in degrees) when $P_b$ is close to $P_z$. \label{subfig:random_vs_component_wise_projections_d}]{%
   \includegraphics[width=0.23 \textwidth]{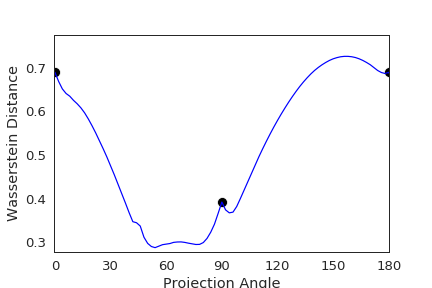}
 }
 \
 \caption{The hash output $b$ (fake, scattered points) versus the prior, target hash codes $z$ (real, part of the corner clusters) in the 2-D discrete space. In Figures~\ref{subfig:random_vs_component_wise_projections_a} and~\ref{subfig:random_vs_component_wise_projections_c} , we want to push the data points closer to the corners. Clearly, projections onto the axes (black vertical and horizontal lines) result in the most correction. In Figures~\ref{subfig:random_vs_component_wise_projections_b} and~\ref{subfig:random_vs_component_wise_projections_d}, we can clearly see this effect.}
 \label{fig:random_vs_component_wise_projections}
\end{figure}

In our experiments, we observe that it is difficult to match $P_b$ to $P_z$ by solving the OT. We conjecture that the reason is because the OT is a poor estimate of the Wasserstein distance. It has high variance when the number of samples $N$ in the mini-batches is small~\cite{iohara2019generative}. On the other hand, SWD is a more ``generalizable'' distance estimate than the OT~\cite{deshpande2019max}. Generalization refers to the number of samples the algorithm needs to converge to the target distribution.


The empirical SWD of two samples $\mathcal{D}$ and $\mathcal{F}$ is estimated as follows:

\begin{equation}
    \hat{W}(\mathcal{D}, \mathcal{F}) = \frac{1}{|N_{\omega}|} \sum_{k=1}^{N_{\omega}} W(\mathcal{D} \omega_k, \mathcal{F} \omega_k)
\end{equation}

where $\omega_k \in R^d$ is a vector which defines the projection onto a random, one-dimensional direction and $N_{\omega}$ is the number of such random projections. While SWD has a better sample complexity than the OT, it needs a high number of random directions. 

In hashing, and especially in matching the latent space of $b$ to the target, discrete prior $P_z$, each sample $z$ of $P_z$ lies at the vertices of the hypercube. Without loss of generality, assume that $z$ is two dimensional, therefore in $\{-1,1\}^2$ (this is similar to sampling $z \in \{0,1\}^2$ where the only difference is the activation function \textit{tanh} instead of \textit{sigmoid} after the logit output of the encoder $f$). A discrete sample $z$ falls into one of the four possible corners, as seen in Figures~\ref{subfig:random_vs_component_wise_projections_a} and~\ref{subfig:random_vs_component_wise_projections_c}. Figures~\ref{subfig:random_vs_component_wise_projections_b} and~\ref{subfig:random_vs_component_wise_projections_d} show the one-dimensional Wasserstein distances for different directions at different angles from the vector $(1, 0)$. In Figure~\ref{subfig:random_vs_component_wise_projections_b}, when the generated data points $b$ are further away from the corners, the projections onto the directions along the axes (those with angles $45r^{\circ}$ for different integer values $r$) have the most distances, thus best describe the separation between the samples of $b$ and $z$. The projections onto other directions will underestimate this separation between the samples of $z$ and $b$. In Figure~\ref{subfig:random_vs_component_wise_projections_d}, when more data points are closer to the corners, the projections onto $0^{\circ}$ or $180^{\circ}$-degree direction still well describe the most distances.  In other words, if we project the data points onto these axes and average the one-dimensional Wasserstein distances along these axes, the resulting Wasserstein distance is a better distance estimate compared to the random projections. Therefore, this motivates us to estimate the Wasserstein distance by averaging the distances along each dimension or $b$ and $z$, as follows:

\vspace{-10pt}
\begin{align}
    L_G = \hat{W}(\mathcal{D}, \mathcal{F}) = \frac{1}{m} \sum_{i}^{m} \hat{W}(\mathcal{D}_i, \mathcal{F}_i)
    \label{eqn:idr_loss1}
\end{align}

where $\mathcal{D}_i = \{z_i^{(1)}, z_i^{(2)},...,z_i^{(N)}\}$ and $\mathcal{F}_i = \{b_i^{(1)}, b_i^{(2)},...,b_i^{(N)}\}$. Solving the OT in the one-dimensional space has a significantly small computational cost~\cite{deshpande2018generative}. The cost of such operation is equivalent to the one-dimensional array sort plus the distance calculation. Therefore, the Wasserstein-p distance $L_G$ can be calculated as follows:

\vspace{-5pt}
\begin{align}
    L_G = \frac{1}{m} \sum_{i}^{m} \sum_{j}^{N} ||z^{(\sigma_i(j))} - b^{(\sigma_i(j))}||_p
    \label{eqn:idr_loss}
\end{align}
where $\sigma_i$ is the sorting operation applied to dimension $i$, and $z^{(\sigma_i(j))}$ and $b^{(\sigma_i(j))}$ are the ranked $j^{th}$ values of the sets $\mathcal{D}_i$ and $\mathcal{F}_i$, respectively. The cost of solving the proposed Wasserstein distance estimate is $O(Nmlog(Nm))$. Since $m$ is fixed and smaller than $N$, this is an order of magnitude faster than the high-dimensional OT's computational cost of $O(N^{2.5}\log(Nm))$. We can also show that the proposed $L_G$ calculation is a valid distance measure and its sample complexity is polynomial.

\begin{theorem}
The proposed Wasserstein-p calculation is a valid distance and it has a polynomial sample complexity.
\end{theorem}

\begin{proof}
Let $\omega_k \in R^m$ be the one-hot vector whose component $k$ is 1. We can see that $\mathcal{D}\omega_k$ and $\mathcal{F}\omega_k$ are projections of the data matrices onto each dimension $k$. By setting $N_{\omega}=m$, this is exactly the formulation of SWD. Thus, by a similar proof in~\cite{deshpande2019max}, it is trivial to show that the proposed Wasserstein estimate is a valid distance and has a polynomial complexity.
\end{proof}



\SetKwProg{function}{}{}{}
\begin{algorithm}[!htb]
\LinesNumberedHidden
\KwIn{
    \text{Training data} $X$, \newline
    \text{Discrete, latent code size } $m$, \newline
    \text{Number of training iterations $K$}. \newline
    Number of reconstruction steps per one adversarial matching step $l$. \newline
    \text{Learning rate $\eta$}.
    }
\KwOut{$\{W_f\}$ parameters of the encoder} 

\For{number of training iterations $K$}
    {
    	    \For{number of reconstruction steps $l$} {
                Sample a minibatch of $N$ examples $\{x_{pixel}^{(1)},...,x_{pixel}^{(N)}\}$.  \newline
                Compute the feature vectors $x^{(j)} = h(x_{pixel}^{(j)})$ for $j=\{1,...,N\}$. \newline
                Compute $L_A = ||g(f(x)) - x||_2^2$. \newline
                Update $W_f \leftarrow W_f - \eta \nabla_{W_f} L_A$. \newline 
                Update $W_g \leftarrow W_g - \eta \nabla_{W_g} L_A$. \newline 
            } 
        	
            Sample a minibatch of $N$ examples $\{x_{pixel}^{(1)},...,x_{pixel}^{(N)}\}$.  \newline
            Compute the feature vectors $x^{(j)} = h(x_{pixel}^{(j)})$ for $j=\{1,...,N\}$. \newline
        	Sample $N$ vectors $\{z^{(1)},...,z^{(N)}\}$ where $z^{(i)} \sim P_z$. \newline
            \function {Compute $L_G$} {
                \For{each dimension $i$} {
                        $L_G(i) = \hat{W}(\{z_i^{(1)}, z_i^{(2)},...,z_i^{(N)}\}, \{b_i^{(1)}, b_i^{(2)},...,b_i^{(N)}\})$
                } 
            Set $L_G = \frac{1}{m}\sum_i^m L_G(i)$ 
            }

            Update $W_f \leftarrow  W_f - \eta \nabla_{W_g} L_G$.  
    }
\caption{\textbf{DCW-AE} Model Training}
\label{alg:IDR}
\end{algorithm}

Unlike SWD which employs a large number of random directions, the proposed estimate uses the directions that best separate the generated $b$ and the real data $z$. Similar to SWD, estimating the Wasserstein distance from the $m$ one-dimensional projections has a polynomial sample complexity. This is an important advantage over the OT estimation, which has an exponential sample complexity. 

The proposed estimate is also related to Max-SWD~\cite{deshpande2019max}. Max-SWD finds the single best projected dimension that best describes the separation of the two samples, by employing the discriminator that classifies real and fake data points. This results in the problematic minimax game. Our proposed calculation can estimate the similar averaged distance without the discriminator. For example, in  Figure~\ref{subfig:random_vs_component_wise_projections_b}, we can show that, given the optimal discriminator, Max-SWD finds the single direction whose distance would be the average of the distances along the $0^{\circ}$-direction and $90^{\circ}$-direction. Our approach will also estimate the same average, but without using the discriminator. 




We call the adversarial autoencoder with the proposed loss calculation \textbf{D}iscrete \textbf{C}omponent-wise \textbf{W}asserstein \textbf{A}uto\textbf{E}ncoder (DCW-AE). The objective function of the DCW-AE can be written as follows:


\begin{align}
	L = L_{A} + L_{G}
    \label{eqn:loss_overall}
\end{align}


We summarize the training algorithm of DCW-AE in Algorithm~\ref{alg:IDR}. While we can have a single minimization step on $L$, we find that alternatively minimizing $L_A$ and $L_G$ works better in practice. (similar to Alternating Least Squares approach traditionally used in matrix factorization). Specifically, we minimize $L_A$ for a few steps $l$ on every minimization step of $L_G$. In all of our experiments, we set $l=5$. Note that this is not a minimax game because $L_A$ and $L_G$ are different losses and are not related through a divergence or a value function, as in WGAN~\cite{arjovsky2017wasserstein} and Jensen-Shannon GAN~\cite{goodfellow2014generative}.

\section{Experiments} \label{experiments}

In this section, we present the experimental results to demonstrate the effectiveness of the proposed hashing method over the existing adversarial autoencoders and other hashing methods.

\subsection{Datasets Used} \label{datasets}
We utilize the following datasets in our performance evaluation experiments:

\begin{table*}[!t]
\caption{Performance comparison of different methods using P@1000. The best P@1000 value for each experiment is in bold.}
\begin{tabularx}{\linewidth}{|l|C|C|C|C|C|C|C|C|}
\hline
\multirow{2}{*}{Method} &
\multicolumn{2}{c|}{\textbf{MNIST}} 
              &
\multicolumn{2}{c|}{\textbf{CIFAR10}}                   & \multicolumn{2}{c|}{\textbf{FLICKR25K}}                 & \multicolumn{2}{c|}{\textbf{PLACE365}}                  \\ \cline{2-9} 
                    & \textbf{$m=32$} & \textbf{$m=64$} & \textbf{$m=32$} & \textbf{$m=64$} &
                    \textbf{$m=32$} & \textbf{$m=64$} &
                    \textbf{$m=32$} & \textbf{$m=64$}\\ \hline
\textbf{LSH}            & 0.3141 & 0.4306 & 0.1721           & 0.1731           & 0.0131           & 0.0188           & 0.0058           & 0.0101                       \\ \hline
\textbf{SpecHash}       & 0.4416 & 0.5042 & 0.1995           & 0.1982                      & 0.0213           & 0.0233                     & 0.0075           & 0.0076                       \\ \hline
\textbf{ITQ}           &  0.5988 & 0.6081 & 0.2424           & 0.2585                      & 0.0264           & 0.0302                     & 0.0088           & 0.0096                     \\ \hline
\textbf{SGH}          & 0.4713 & 0.5224  & 0.1672           & 0.1803                      & 0.0130          & 0.0137                     & 0.0061           & 0.0061                      \\ \hline
\textbf{SSDH}       & 0.4417 & 0.4698  & 0.2218           & 0.1766                     & 0.0260           & 0.0271                     & 0.0149           & 0.0160                      \\ \hline
\textbf{DistillHash}          & 0.4817 & 0.4917 & 0.2319	&	0.2331	&	0.0261	&	0.0283	&	0.0150	&	0.0165	          \\ \hline
\textbf{WGAN-AE}    & 0.5961 &  0.5972   & 0.1948           & 0.2170           & 0.0219           & 0.0269               & 0.0158           & 0.0174                     \\ \hline
\textbf{OT-AE}      & 0.6082 &  0.6117   & 0.2370           & 0.2406           & 0.0222           & 0.0273               & 0.0147           & 0.0189                      \\ \hline
\textbf{DCW-AE}     & \textbf{0.6451} &  \textbf{0.6545}     & \textbf{0.2692}  & \textbf{0.2747}   & \textbf{0.0282}  & \textbf{0.0336}  & \textbf{0.0205}          & \textbf{0.0259}          \\ \hline
\end{tabularx}
\label{tbl:precisions}
\end{table*}

\begin{itemize}[leftmargin=*]
    \item \textbf{MNIST}~\footnote{http://yann.lecun.com/exdb/mnist/}: This dataset consists of of 70,000 digit images. We randomly select 10,000 images as the query set and the remaining images for the training and retrieval sets. 
    \item \textbf{CIFAR10}~\cite{krizhevsky2009learning}: This dataset consists of 60,000 natural images categorized uniformly into 10 labels. We randomly select 1,000 images from each label for the query set and use the remaining images for the training and retrieval sets. Hence, the query set contains 10,000 images and the training/retrieval set contains the same 50,000 images.
    \item \textbf{FLICKR25K}~\cite{bui2018sketching}: This dataset consists of 25,000 social photographic images downloaded from Flickr~\footnote{https://www.flickr.com/}. There are a total of 250 different class labels. We randomly select 20 images from each label for the query set and similarly use the remaining images for the training and retrieval sets. The final query dataset contains 5,000 images and the training/retrieval set contains the same 20,000 images.
    \item \textbf{PLACE365}~\footnote{http://places2.csail.mit.edu/}: This dataset consists of 1.8 millions of scenery images organized into 365 categories (labels). We randomly select 10 images from each label for the query set and 500 images from each label for the trainin and retrieval sets. The final query dataset contains 3,650 images and the training/retrieval set contains 182,500 images. 
\end{itemize}

\subsection{Evaluation Metrics}
For evaluating the performance of the proposed model, we follow the standard evaluation mechanism that is widely accepted for the problem of image hashing
- the \textbf{precision@R} (\textbf{P@R}) and \textbf{mean average precision} (\textbf{MAP}). Given the query images, P@R and MAP are calculated as follows:

\begin{align} 
    \text{Precision}(R, q) &= \frac{\sum_{r=1}^R \delta(r,q)}{R} \\
    \text{P@$R$} &= \frac{1}{Q}\sum_{q=1}^Q  \text{Precision}(R, q) \\
    \text{AP}(q) &= \frac{1}{N_q} \sum_{r=1}^N \text{Precision}(r, q) \times \delta(r,q)
    \\
    \text{MAP} &= \frac{1}{Q} \sum_{q=1}^Q AP(q),
\end{align}
where $N$ is the size of the retrieval set, $R$ is the number of retrieved images, $N_q$ is the number of all relevant images in this set, $Q$ is the size of the query set and $\delta(r,q) = 1$ only when the $r$-th retrieved image is relevant to the query image $q$; otherwise $\delta(r,q) = 0$. A retrieved image is relevant if its ground-truth label is the same as the label of the query image. 

\subsection{Comparison Methods}

We compare the performance of the proposed method with various representative unsupervised image hashing methods.

\begin{itemize}[leftmargin=*]
    \item \textit{Locality Sensitive Hashing } (\textbf{LSH})~\cite{leskovec2014mining}: a widely-used data-independent, shallow hashing method using random projection. 
    \item \textit{Spectral Hashing} (\textbf{SpecHash})~\cite{weiss2009spectral}: an unsupervised shallow hashing method whose goal is to preserve locality while finding balanced, uncorrelated hashes by solving the Eigenvector problem.
    \item \textit{Iterative Quantization} (\textbf{ITQ})~\cite{gong2013iterative}: the state-of-the-art shallow hashing method that alternately minimizes the quantization error to achieve better hash codes.
    \item \textit{Stochastic Generative Hashing} (\textbf{SGH})~\cite{dai2017stochastic}: a representative hashing method that, similar to the proposed method, also minimizes the reconstruct loss in an autoencoder model.  
    \item \textit{Semantic Structure-based Deep Hashing} (\textbf{SSDH})~\cite{yang2018semantic}: an unsupervised deep hashing method that learns the hash function by preserving heuristically-defined semantic structure of the data. The semantic structure is extracted from a pre-trained neural network (such as VGGNet).
    \item \textit{Deep Hashing by Distilling Data Pairs} (\textbf{DistillHash})~\cite{yang2018semantic}: the state-of-the-art unsupervised deep hashing method that is, in principle, similar to SSDH. However, the semantic structure is constructed by distilling data pairs that are consistent with the Bayes optimal classifier.
    \item \textit{Wasserstein Adversarial Autoencoder} (\textbf{WGAN-AE}): adversarial autoencoder model for hashing which employs the critic that estimates the Wasserstein from the dual domain. This is an improved version of the adversarial autoencoder defined in~\cite{doan2019adversarial}.
    \item \textit{OT-Wasserstain Adversarial Autoencoder} (\textbf{OT-AE}): the adversarial autoencoder model for hashing which directly minimizes the Wasserstein distance using the OT formulation in the primal domain, (discussed in Section~\ref{sec:ot_ae}).
    \item \textit{The proposed method} (\textbf{DCW-AE})
\end{itemize}

\begin{figure*}[!ht]
 \centering
 \renewcommand{\thesubfigure}{a}
 \subfloat[MNIST: 32 bits \label{subfig-1:pr_32bits}]{%
   \includegraphics[width=0.32 \textwidth]{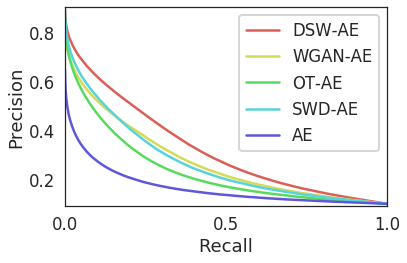}
 }
 \hfill
 \renewcommand{\thesubfigure}{b}
 \subfloat[CIFAR10: 32 bits \label{subfig-2:pr_64bits}]{%
   \includegraphics[width=0.32 \textwidth]{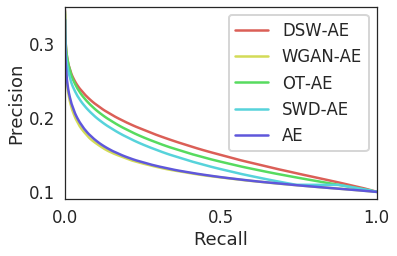}
 }
  \hfill
 \renewcommand{\thesubfigure}{c}
 \subfloat[FLICKR25K: 32 bits \label{subfig-2:pr_64bits}]{%
   \includegraphics[width=0.32 \textwidth]{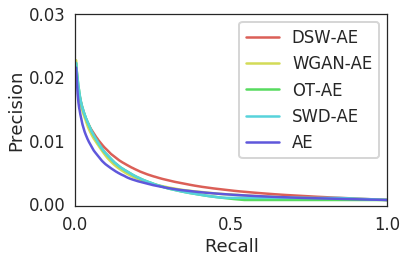}
 }
 \vfill
  \renewcommand{\thesubfigure}{d}
 \subfloat[MNIST: 64 bits \label{subfig-1:pr_32bits}]{%
   \includegraphics[width=0.32 \textwidth]{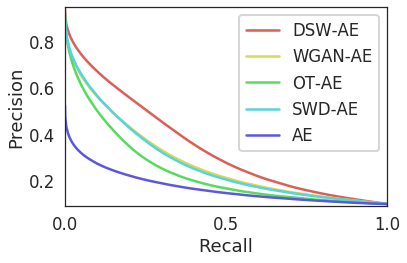}
 }
 \hfill
 \renewcommand{\thesubfigure}{e}
 \subfloat[CIFAR10: 64 bits \label{subfig-2:pr_64bits}]{%
   \includegraphics[width=0.32 \textwidth]{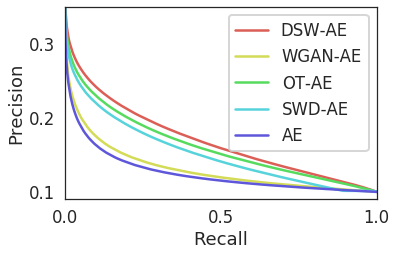}
 }
  \hfill
  \renewcommand{\thesubfigure}{f}
 \subfloat[FLICKR25K: 64 bits \label{subfig-2:pr_64bits}]{%
   \includegraphics[width=0.32 \textwidth]{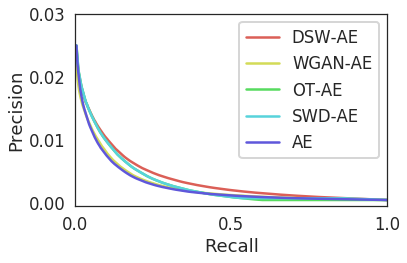}
 }
 \caption{Performance comparison of various methods using Precision-Recall curves of code lengths of 32 and 64 bits.}
 \label{fig:pr_curves}
\end{figure*}


\begin{table}[!t]
\caption{Performance comparison of different methods using MAP for $m=64$. The best MAP values are shown in bold.}
\begin{tabularx}{\columnwidth}{|l|C|C|C|C|}
\hline
Method            & {\footnotesize\textbf{MNIST}} & {\footnotesize\textbf{CIFAR10}} & {\footnotesize\textbf{FLICKR25K}} & {\footnotesize\textbf{PLACE365}} \\ \hline
\textbf{LSH}       & 0.2228 & 0.1477           & 0.0348             & 0.0101            \\ \hline
\textbf{SpecHash}  & 0.2856 &  0.1265           & 0.0535             & 0.0121            \\ \hline
\textbf{ITQ}      & 0.3121 & 0.1824           & 0.0535             & 0.0163            \\ \hline
\textbf{SGH}      & 0.2912 & 0.1372           & 0.0196             & 0.0052            \\ \hline
\textbf{SSDH}     & 0.2834 & 0.1664           & 0.0504             & 0.0107            \\ \hline
\textbf{DistillHash}    & 0.3071  & 0.1831	         & 0.0579	          &	0.0171            \\ \hline
\textbf{WGAN-AE}  & 0.3033 & 0.1775           & 0.0642             & 0.0161            \\ \hline
\textbf{OT-AE}   & 0.3431 & 0.1777           & 0.0658             & 0.0165            \\ \hline
\textbf{DCW-AE}    & \textbf{0.3901} & \textbf{0.2084}           & \textbf{0.0790}             & \textbf{0.0184}            \\ \hline
\end{tabularx}
\label{tbl:MAPs}
\end{table}

\begin{table*}[!t]
\centering
\begin{tabular}{m{0.5in} m{1.0in} m{4.5in}}
\center{\textbf{$m$}} &
\center{\textbf{Query Image}} & \textbf{\quad \quad \quad \quad \quad \quad \quad \quad \quad \quad \quad \quad Top 10 Retrieved Images}\\ \toprule
\center{32} & \center{\includegraphics[height=0.45in]{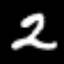}} &
\includegraphics[height=0.45in]{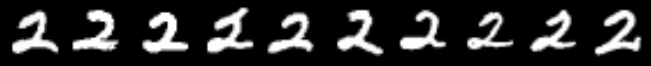} \\ 
\center{16} & \center{\includegraphics[height=0.45in]{figures/mnist/mnist-query-5.png}} & \includegraphics[height=0.45in]{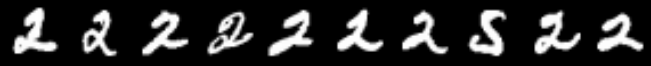} \\
\midrule

\center{32} & \center{\includegraphics[height=0.45in]{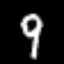}} &
\includegraphics[height=0.45in]{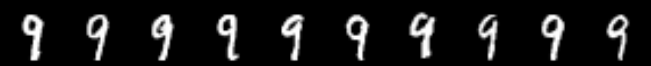} \\ 
\center{16} & \center{\includegraphics[height=0.45in]{figures/mnist/mnist-query-2.png}} & \includegraphics[height=0.45in]{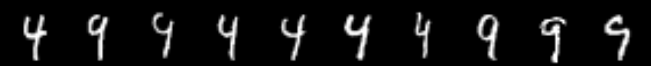} \\
\bottomrule

\end{tabular}
\caption{An illustration of the top-10 retrieved MNIST digits for a given query image using code length ($m$) of 32 and 16 bits.}
\label{tbl:queried_images_mnist}
\end{table*}


\noindent  \textbf{Implementation Details:} For the existing shallow hashing techniques, we employ a pre-trained VGG network and extract the pooling-\textit{fc7} feature vectors of the images~\cite{simonyan2014very}. For our model, we employ VGG for the feature extractor $h$ on CIFAR10, FLICKR25K and PLACE365. For MNIST, we do not use the feature extractor $h$ and directly learn to reconstruct the pixel image, i.e. $x = x_{pixel}$. The encoder/decoder are multi-layer perceptrons (MLP) with hidden layers as $1000\rightarrow1000\rightarrow500$ and $500\rightarrow1000\rightarrow1000$, respectively. We implement our method in pyTorch~\footnote{http://pytorch.org/} and train our model using Stochastic Gradient Descent (SGD) along with Adam optimizer~\cite{kingma2014adam}. We use a mini-batch size of 128 examples for DCW-AE and WGAN-AE. For OT-AE, we try different mini-batch sizes ranging from 128 to 512. For a fair and strict evaluation, we perform a grid search to find the best hyper-parameters for each of the methods; and report averaged results over five runs (three runs for PLACE365 dataset). The source code and the datasets used in our experiments will be made publicly available on a Github repository upon the acceptance of this paper.

\subsection{Performance Results}

\begin{figure*}[!t]
 \centering
 \renewcommand{\thesubfigure}{a}
 \subfloat[ITQ \label{subfig-1:pr_32bits}]{%
   \includegraphics[width=0.23 \textwidth]{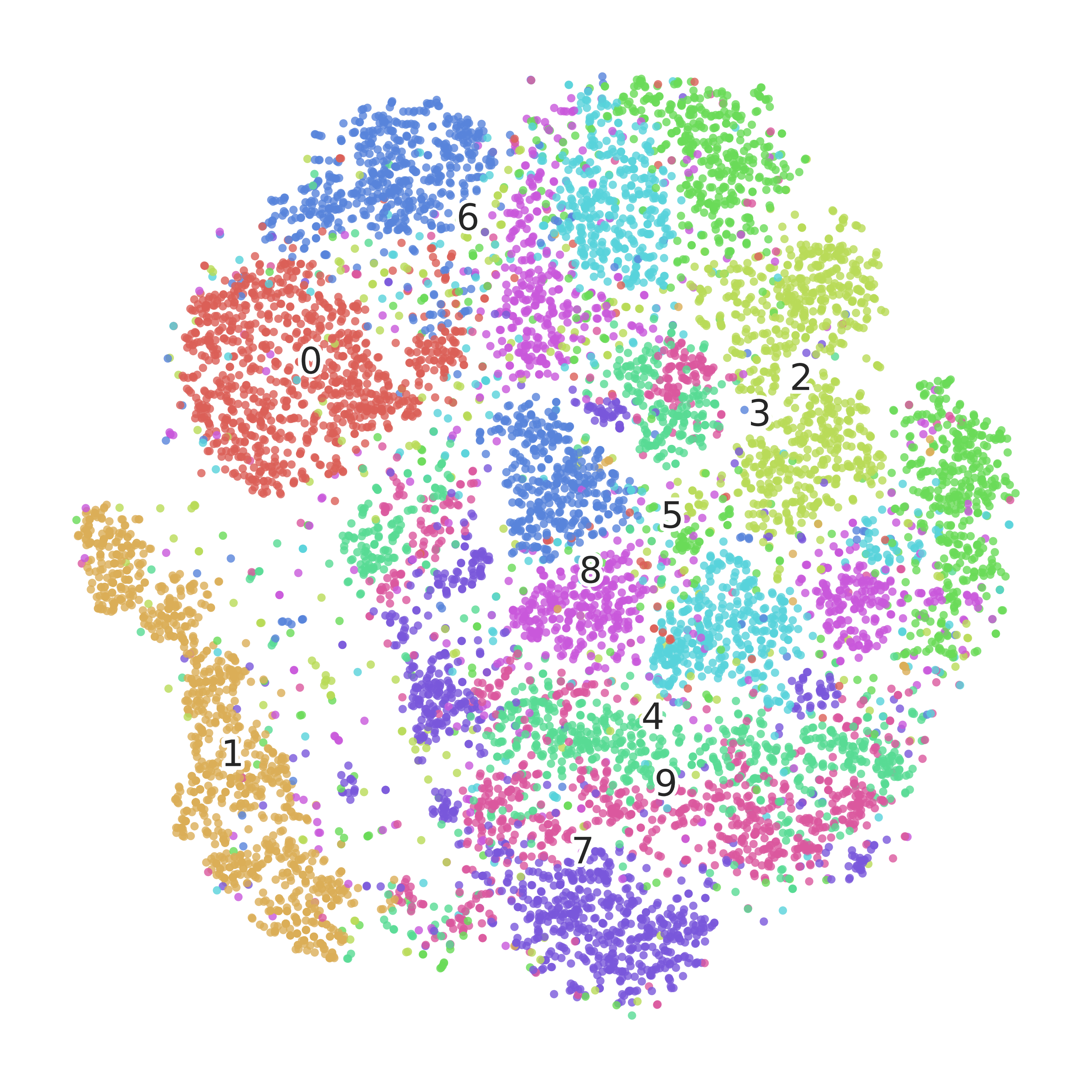}
 }
 \hfill
  \renewcommand{\thesubfigure}{b}
 \subfloat[DistillHash \label{subfig-1:pr_32bits}]{%
   \includegraphics[width=0.23 \textwidth]{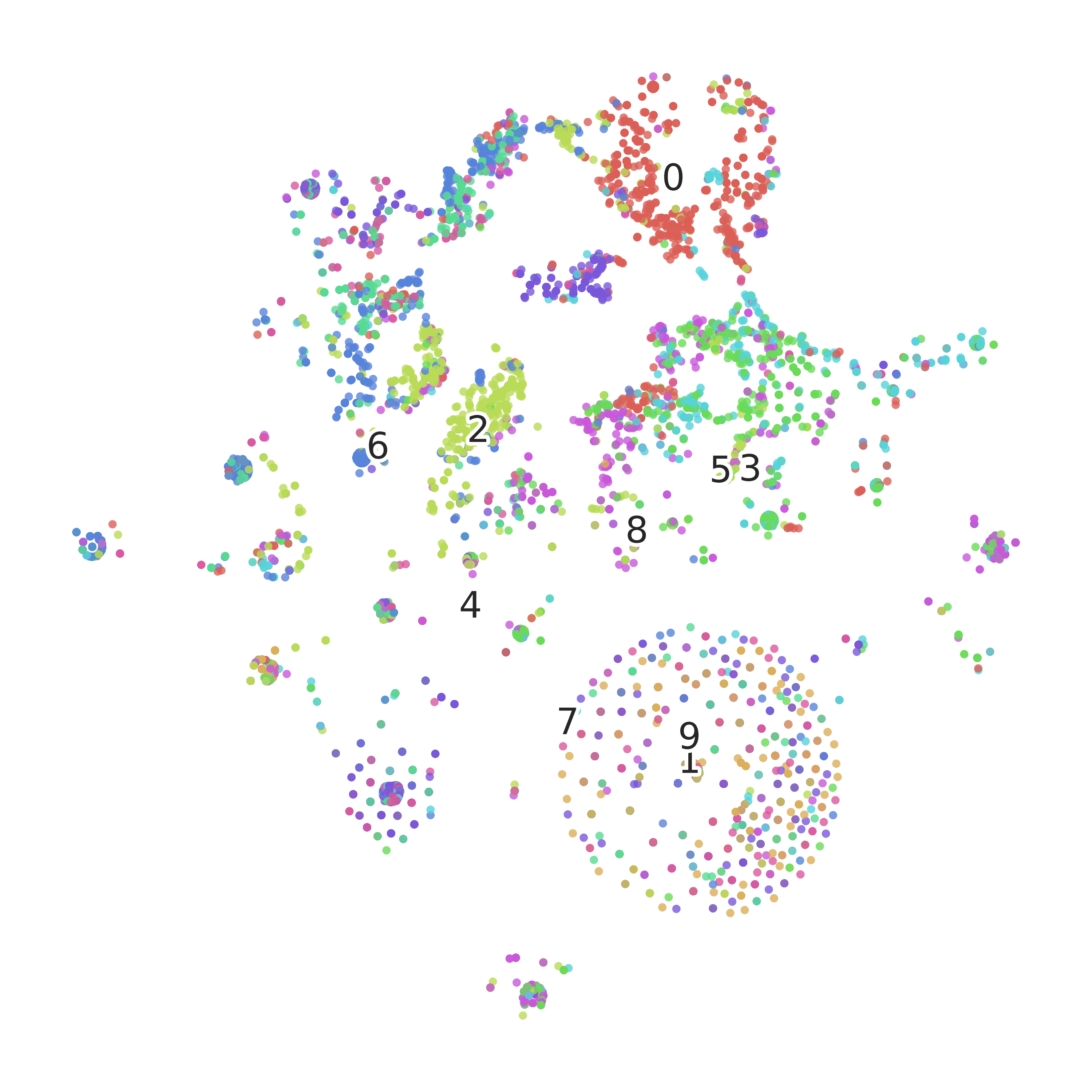}
 }
 \hfill
 \renewcommand{\thesubfigure}{c}
 \subfloat[OT-AE \label{subfig-1:pr_32bits}]{%
   \includegraphics[width=0.23 \textwidth]{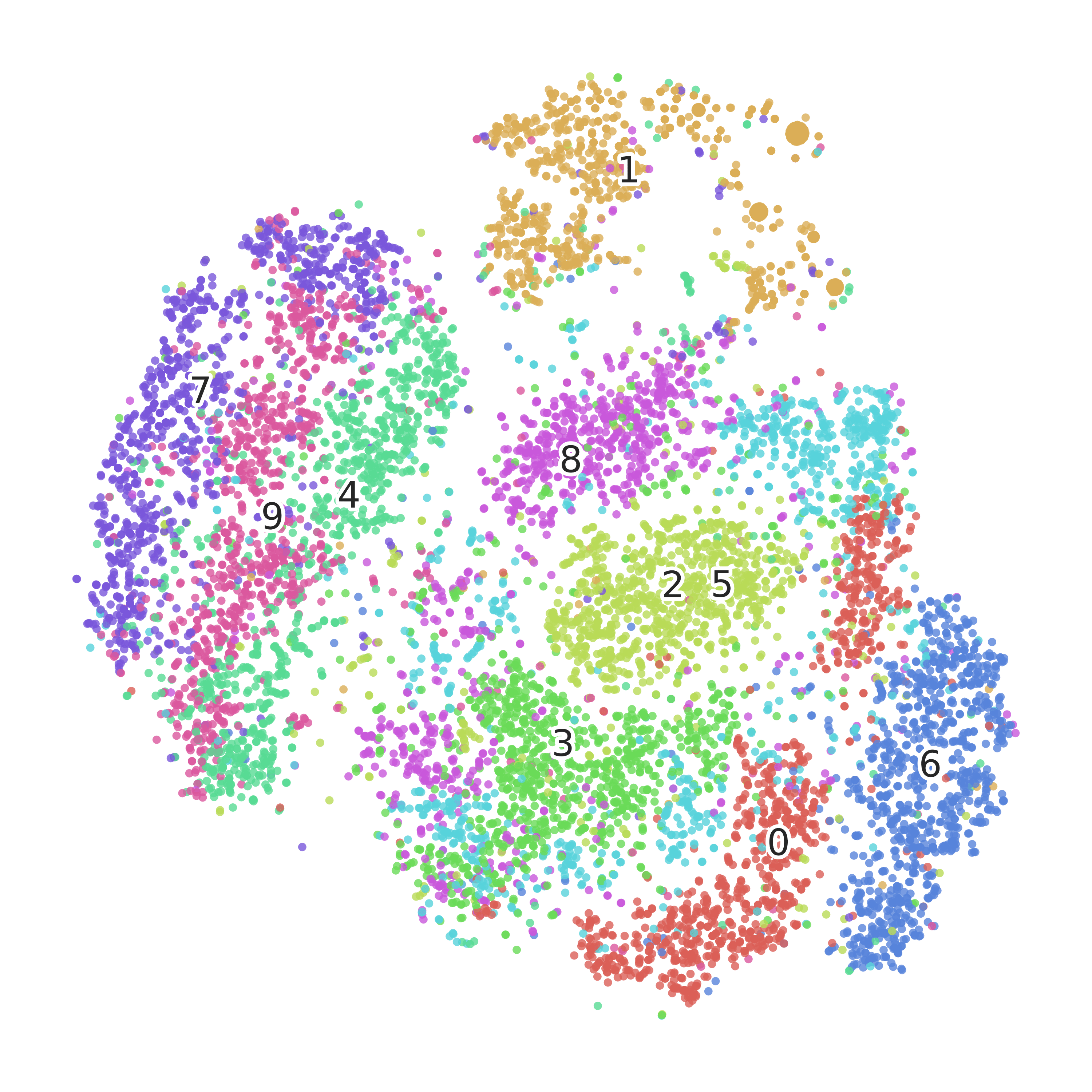}
 }
 \hfill
 \renewcommand{\thesubfigure}{d}
 \subfloat[DCW-AE \label{subfig-2:pr_64bits}]{%
   \includegraphics[width=0.23 \textwidth]{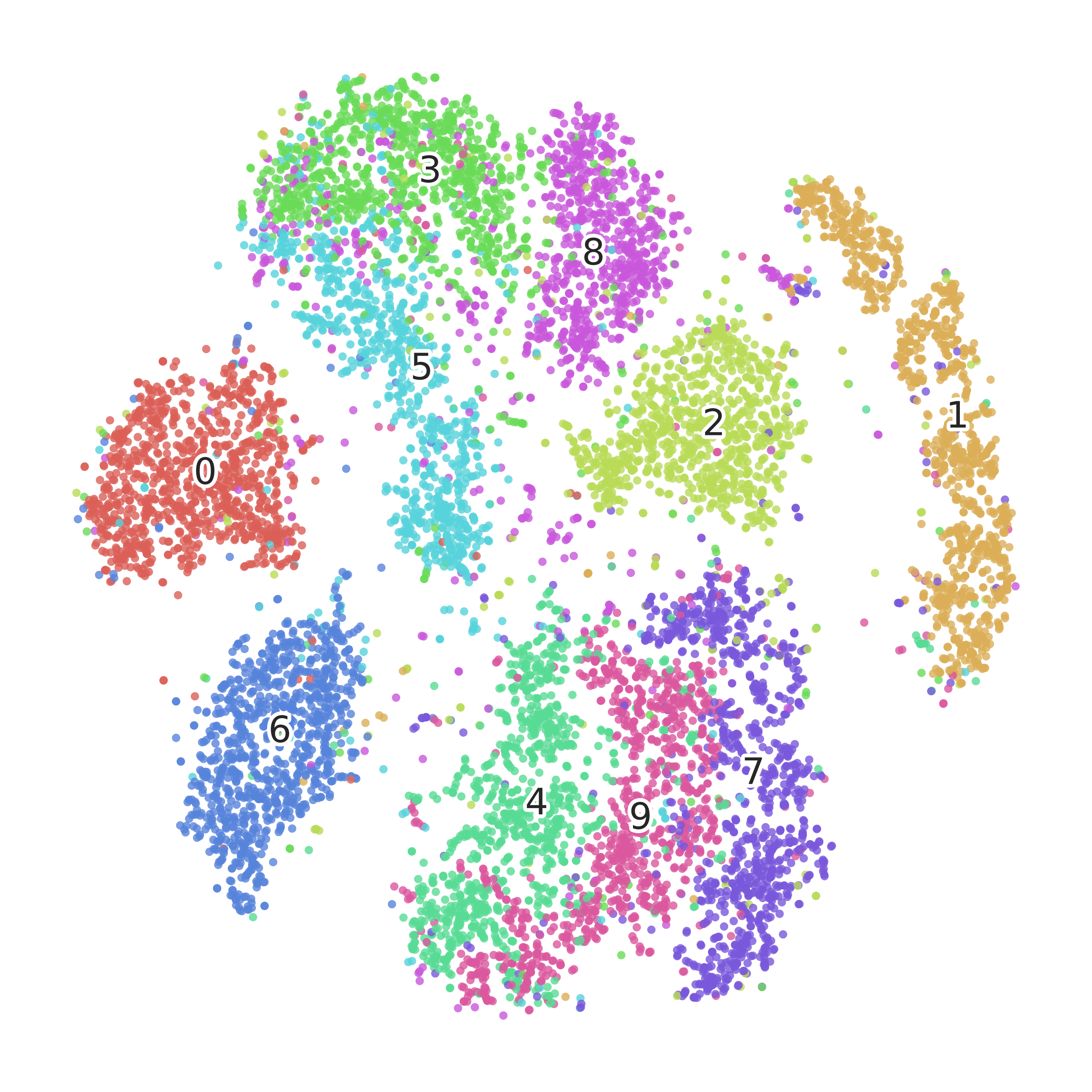}
 }
 \caption{T-SNE embedding of the generated discrete hash codes on the MNIST dataset. The digit identities are color-coded.}
 \label{fig:tsne_embedding}
\end{figure*}
In this experiment, we measure the performance of various methods for the image retrieval task. Table \ref{tbl:precisions} shows the P@1000 results across different lengths of the hash codes. DCW-AE consistently outperforms all the baseline methods at different lengths of the hash codes. Specifically, DCW-AE has a relative performance improvement of more than 10\% in CIFAR10 and FLICKR25K. Similarly, in Table~\ref{tbl:MAPs}, we report the MAP results for all the methods at different lengths of the hash codes. Again, DCW-AE consistently achieves the best MAP results. The improvements of our method over the baselines are statistically significant according to the corresponding paired t-tests (p-value $< 0.01$). 

The P@1000 and MAP results demonstrate the superiority of our method compared to various state-of-the-art approaches for the image hashing problem. One important result is the improvement in performance of DCW-AE compared to OT-AE. This supports our claim that the existing adversarial autoencoders cannot learn the optimal hash function compared to our DCW-AE model. This demonstrates the significance of a generalizable Wasserstein estimate. A lower sample-complexity estimate makes it easier for the algorithm to converge to the target distribution using mini-batch training algorithms.
\vspace{-0.09in}
\subsection{Ablation Study}

We further evaluate the effectiveness of the proposed adversarial learning procedure using an ablation study. Figure~\ref{fig:pr_curves} shows the Precision-Recall curves of the different Autoencoder-based hashing models. AE denotes the vanilla Autoencoder without any adversarial regularization. SWD-AE is the Adversarial Autoencoder whose adversarial matching cost is SWD. We observe that all adversarial-based Autoencoder models outperform AE. This demonstrates the importance of the adversarial learning for Autoencoders in image hashing. Furthermore, replacing the dual Wasserstein estimate (in WGAN-AE) and the OT estimate (in OT-AE) with our proposed estimate further improves the retrieval performance. DCW-AE's performance is significantly better than that of OT-AE and SWD-AE. We hypothesize that this improvement is primarily due to the following reasons:

\begin{itemize}[leftmargin=*]
    \item The proposed distance estimate of DCW-AE has a better sample complexity (i.e., more generalizable) than OT. Better sample complexity allows mini-batch optimization techniques such as SGD to converge better to the target distribution.
    \item The proposed distance estimate of DCW-AE is a better distance compared to SWD. Our distance estimate induces a weaker topology than SWD~\cite{arjovsky2017wasserstein}. In other words, our estimate makes it easier to converge to the target distribution. 
\end{itemize}

\subsection{Qualitative Evaluation}

In Table~\ref{tbl:queried_images_mnist}, we show the top-10 retrieved digit images of two query images corresponding to digits 3 and 9. DCW-AE method has successfully retrieved relevant digits; when the retrieved images are false positives, we can still see that they contains similar appearances (e.g. some handwritten digit 4's are similar to the digit 9). As expected, when increasing the size of the binary code) ($m$), the model makes fewer mistakes. 

Qualitatively, we can visually compare the quality of the hash codes generated by DCW-AE, OT-AE and two best non-adversarial hashing baselines, ITQ and DistillHash. Figure~\ref{fig:tsne_embedding} shows the two-dimensional t-SNE embeddings~\cite{maaten2008visualizing} of the generated hash codes on the query set. In this example, the similarity matrix of SSDH is constructed directly from the image-pixel space. Notice that SSDH generates unreliable hash codes. This shows that, without a reliable construction of the similarity matrix, the retrieval performance of both SSDH and DistillHash is significantly deteriorated. On the other hand, DCW-AE learns a very efficient discrete embedding of the original data; for example, DCW-AE can even separate the most similar digits 9 and 4.


\subsection{DCW-AE's Wasserstein Estimation}

\begin{figure}[!h]
 \centering
 \renewcommand{\thesubfigure}{a}
 \subfloat[Distance Estimates \label{subfig:wasserstein_estimates}]{%
   \includegraphics[width=0.23 \textwidth]{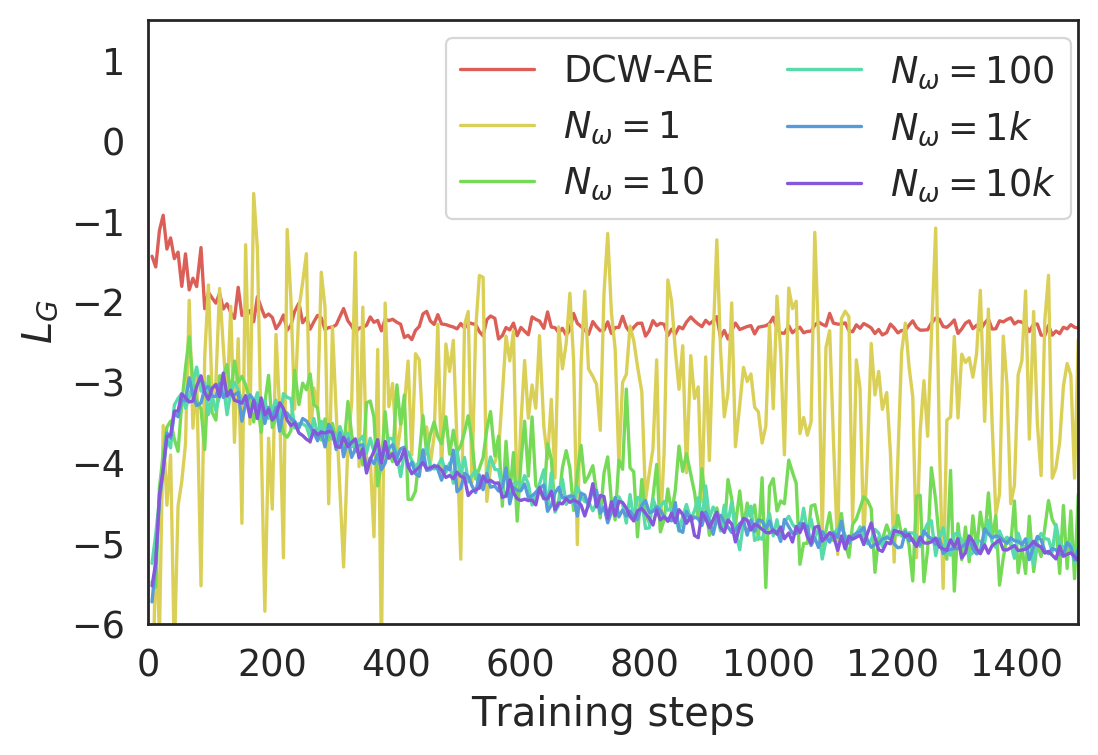}
 }
 \hfill
  \renewcommand{\thesubfigure}{b}
 \subfloat[Gradient Norm \label{subfig:gradnorm}]{%
   \includegraphics[width=0.23 \textwidth]{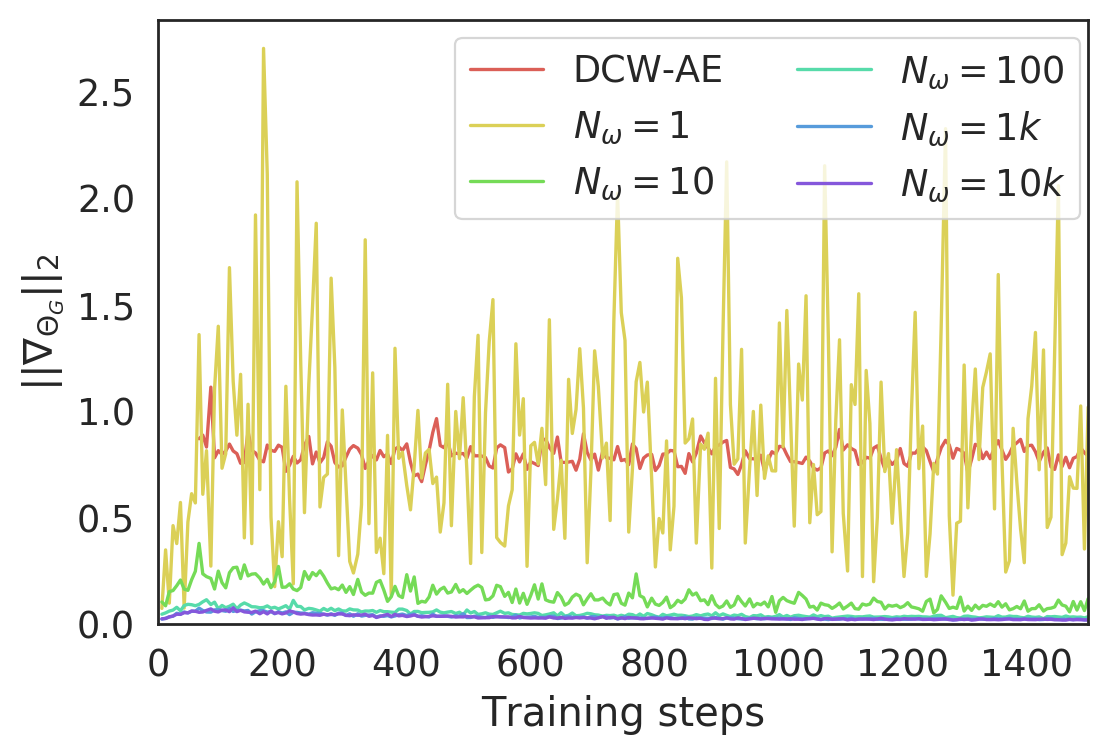}
 }
\caption{Wasserstein distance estimates (left) and the gradient norm of the encoder's parameters (right) during training of DCW-AE and SWD (with varying $N_{\omega}$ values).}
\label{fig:training_estimates}
\end{figure}

In this Section, we show the efficiency of the proposed Wassertein estimate and compare it with SWD. Figure~\ref{fig:training_estimates} shows the distance estimates and the gradient norm of the parameters ($||W_f||_2$) during training of the proposed calculation in DCW-AE and of the original SWD on the CIFAR-10 dataset. The SWD is estimated with different number of random projections $N_{\omega}$. On the extreme case, when $N_{\omega}=1$, both the distance and the gradient fluctuates siginificantly during training. This makes adversarial training become very unstable. When $N_{\omega}$ is higher, we can observe that the distance estimates of SWD are lower. This is because the one-dimensional Wasserstein distances across several random directions are very small and do not contain useful signal for training (see the corresponding gradient). Even when increasing $N_{\omega}$ from 10 to 10K, the estimate is not generally better. On the other hand, the estimates of DCW-AE are higher and its gradient is also more stable. This is due to the fact that the proposed distance estimate averages the distances of directions along which the $P_z$ and $P_b$ are most dissimilar.

\subsection{Computational efficiency of DCW-AE}
\begin{figure}[!h]
\centering
\includegraphics[width=3.2in]{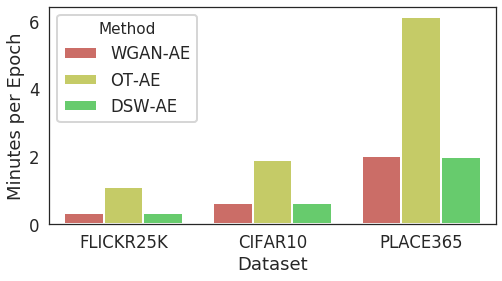}
\caption{Comparison of training times (in minutes) per one epoch on various datasets (using a mini-batch size of 128 examples).}
\label{fig:training_time}
\end{figure}

In this experiment, we compare the training time of the proposed method and the existing adversarial hashing methods. The training time of WGAN-AE, OT-AE and DCW-AE are shown in Figure~\ref{fig:training_time} for three datasets, CIFAR10, FLICKR25K and PLACE365. We report the average training time per epoch. In Figure~\ref{fig:training_time}, the training time of DCW-AE is significantly reduced compared to the training of OT-AE. 

\section{Conclusion} \label{conclusions}

We proposed a novel adversarial autoencoder model for the image hashing problem. Our model learns hash codes that preserves the locality information in the original data. Our model has a much better generalization property than the existing adversarial approaches, thus is able to achieve significant performance gains. Furthermore, the proposed model trains significantly faster than the existing Wasserstein-based adversarial autoencoders. Our experiments validate that the proposed hashing method outperforms all the existing state-of-the-art image hashing methods. Our work makes one leap towards leveraging an efficient, robust adversarial autoencoder for the image hashing problem and we envision that our model will serve as a motivation for improving other adversarially-trained hashing models. The code accompanying this paper is available on Github\footnote{\url{https://github.com/khoadoan/adversarial-hashing}}.




\balance
\bibliographystyle{ACM-Reference-Format}
\bibliography{references}


\end{document}